\documentclass[5p,twocolumn]{elsarticle1}

\usepackage{flushend}
\usepackage{indentfirst}
\usepackage{array}
\usepackage{natbib}
\usepackage{graphicx}
\usepackage{epstopdf}
\usepackage{color}
\usepackage[cmex10]{amsmath}
\usepackage{amsfonts,amssymb,amsxtra,balance}
\usepackage{bm}
\usepackage{mathrsfs}
\usepackage{calc,ifthen}
\usepackage{mathrsfs}
\usepackage{amssymb}

\usepackage{subfig}
\usepackage{caption}

\usepackage{tikz}
\usepackage{circuitikz}
\newtheorem{proposition}{Proposition}
\newtheorem{proof}{Proof}

\newtheorem{lemma}{Lemma}
\newtheorem{remark}{Remark}

\def\begcen{\begin{center}}
\def\endcen{\end{center}}

\newcommand{\col}{\mbox{col}}

\def\diag{\mbox{diag}}

\def\calc{{\cal C}}

\def\hal{\frac{1}{2}}

\def\L2{{\cal L}_2}
\def\L2e{{\cal L}_{2e}}

\def\rea{\mathbb{R}}

\def\diag{\mbox{diag}}

\def\begequarr{\begin{eqnarray}}
\def\endequarr{\end{eqnarray}}
\def\begequarrs{\begin{eqnarray*}}
\def\endequarrs{\end{eqnarray*}}
\def\begarr{\begin{array}}
\def\endarr{\end{array}}
\def\begequ{\begin{equation}}
\def\endequ{\end{equation}}
\def\begequs{\begin{equation*}}
\def\endequs{\end{equation*}}
\def\lab{\label}
\def\begdes{\begin{description}}
\def\enddes{\end{description}}
\def\begenu{\begin{enumerate}}
\def\begite{\begin{itemize}}
\def\endite{\end{itemize}}
\def\endenu{\end{enumerate}}

\def\lef[{\left[\begin{array}}
\def\rig]{\end{array}\right]}

\def\begcen{\begin{center}}
\def\endcen{\end{center}}
\def\begrem{\begin{remark}\rm}
\def\endrem{\end{remark}}

\usepackage{color}

\begin{document}

\begin{frontmatter}
\title{Energy Shaping Control of an Inverted Flexible Pendulum Fixed to a Cart}
\author[iitb]{Prasanna S. Gandhi}\ead{gandhi@me.iitb.ac.in}
\author[supelec]{Pablo Borja}\ead{luisp.borja@lss.supelec.fr}
\author[supelec]{Romeo Ortega}\ead{ortega@lss.supelec.fr}

\address[supelec]{Laboratoire de Signaux et Syst\`emes, CentraleSupelec, 91192 Gif-sur-Yvette, France}
\address[iitb]{Suman Mashruwala Advanced Microengineering Laboratory, Department of Mechanical Engineering, Indian Institute of Technology.
      Powai, Mumbai 400076, India.}

%\maketitle

%\hfill \today
%%% ----------------------------------------------------------------------

\begin{abstract}
Control of compliant mechanical systems is increasingly being researched for several applications including flexible link robots and ultra-precision positioning systems. The control problem in these systems is challenging, especially with gravity coupling and large deformations, because of inherent underactuation and the combination of lumped and distributed parameters of a nonlinear system. In this paper we consider an ultra-flexible inverted pendulum on a cart and propose a new nonlinear energy shaping controller to keep the pendulum at the upward position with the cart stopped at a desired location. The design is based on a model, obtained via the constrained Lagrange formulation, which  previously has been validated experimentally.  The controller design consists of a partial feedback linearization step followed by a standard PID controller acting on two passive outputs. Boundedness of all signals and (local) asymptotic stability of the desired equilibrium is theoretically established. Simulations and experimental evidence assess the performance of the proposed controller.
\end{abstract}

\end{frontmatter}
\section{Introduction}
\lab{sec1}
%%%%%%%%%%%%%%%%%%%%%%%
%
The problem of stabilization of underactuated mechanical systems, both in the domain of ordinary and partial differential equations, has been widely addressed by several control researchers in recent years. In the domain of flexible mechanisms and robots, flexibility in the links is the main source of underactuation. If the deformations due to flexibility are small it is possible to use an {\em unconstrained} Lagrange formulation and invoke the  Assumed Modes Method (AMM) \cite{amm} to obtain a simple, finite--dimensional model---see ~\cite{dsurvey}  for a recent literature review. This modeling procedure, however, is inapplicable for systems with large deformations, for which a {\em constrained} Euler--Lagrange (EL) formulation is required.  This approach has been adopted in \cite{ojas} to derive an accurate model for a single ultra--flexible link fixed to a cart. Potential energy change owing to ultra--large deformations in the presence of gravity is considered in \cite{ojas} using the constant length of the beam as a holonomic constraint. For a survey on recent control techniques for this class of systems see \cite{ojas,leadlag,fem}.

The objective of this paper is to design an energy shaping controller with {\em guaranteed stability properties} for the model of a single ultra--flexible link fixed to a cart reported in \cite{ojas}. As is well known  \cite{ORTDONROM} the application of energy shaping controllers is stymied by the need to solve partial differential equations (PDEs) that identify the mechanical structure (Lagrangian or Hamiltonian)  that is assigned to the closed--loop. To propose a truly constructive energy shaping scheme, that does not require the solution of PDEs, it was recently proposed in \cite{DONetal} to relax the constraint of preservation in closed--loop of the EL structure. The design in \cite{DONetal} proceeds in two steps, first, we apply a partial feedback linearization (PFL) \cite{spong} that transforms the system into Spong's normal form---if this system is still EL, two new passive outputs are immediately identified. Second, a classical PID around a suitable combination of these passive outputs completes the design.

It is shown in the paper that this technique, developed for standard EL systems in \cite{DONetal}, is also applicable to the constrained EL system at hand. This extension is far from obvious, because the (lower order) dynamics that results from the projection of the system on the manifold defined by the constraint {\em is not} an EL system. In spite of this fact it is shown that, because of the workless nature of the forces introduced by the constraints, it is still possible to identify the two new passive outputs to which the PID is applied.

The remainder of the paper is organized as follows. Section \ref{sec2} presents the full constrained EL dynamics of the system and its reduced order projection. Section \ref{sec3} presents the proposed energy shaping control algorithm. Section \ref{sec4} presents the simulation results, while in Section \ref{secexp} we show the experimental ones. Section \ref{sec5} summarizes the work and outlines some future research. \\

%%%%%%%%%%%%
\noindent {\bf Notation:} Unless indicated otherwise, all vectors in the paper are {column} vectors.  Given $n \in \mathbb{N}$, $e_i \in \rea^n$ is the $i$--th Euclidean basis vector of $\rea^n$. For $x \in \rea^n$,  we denote $|x|^2:=x^\top x$. To simplify the expressions, the arguments of all mappings---that are assumed smooth---will be explicitly written only the first time that the mapping is defined. For a scalar function $V \colon \mathbb{R}^{n} \to \mathbb{R}$, we define $\nabla_{x} V := {\left({\partial V \over \partial x}\right)^{\top}}$ and  $\nabla^2_{x} V := {{\partial^2 V \over \partial x^2}}$---when clear from the context the subindex in $\nabla$ will be omitted.
%%
%%% ----------------------------------------------------------------------
\section{System Dynamics and Problem Formulation}
\lab{sec2}
%%%%%%%%%%%%%%%%%%%%%%%
%
In \cite{ojas} a dynamic model that accurately describes the behaviour of the single ultra--flexible link fixed to a cart depicted in Fig. \ref{fig1} is reported. The main feature of this model, which distinguishes it from other models, is that to take into account large deformations of the link its length  is assumed {\em constant}---giving rise to a holonomic constraint. The model is rigorously developed using a constrained EL formalism, combined with a standard application of the AMM, and its validity is experimentally validated.  In this section we present this model, first, in its constrained EL form and then in a reduced form---obtained via the elimination of the constrained equations.
\subsection{Constrained Euler--Lagrange Model}
\label{subsec21}
%%%%%%%%%%%%%
%
The model reported in \cite{ojas} admits a constrained EL representation of the form
\begin{eqnarray}
D(q) \ddot q + C(q,\dot q) \dot q +  B(q)+ R\dot q &=& e_3 \tau + \lambda A(q)\nonumber \\
\Gamma(q)&=&0,
 \label{EL}
\end{eqnarray}
where $q=\col(\theta,x_e,z) \in   \mathbb{D} \times \rea_+ \times \rea$ are the generalised coordinates, $R\geq 0$  is a matrix of damping coefficients. $D> 0$ is the inertia matrix, $C \dot q$ are the Coriolis and centrifugal forces, $B$ is a conservative force vector due to potential energy, $ \tau$ is control vector, $\lambda A$ is a vector of virtual forces due to the holonomic constraint, with $\lambda$ the Lagrange multiplier, and $\Gamma$ is the (constant length) constraint function given by
\begin{equation}
 \Gamma(q):=\int_{0}^{x_e}\sqrt{1+\left[\theta\phi'(x)\right]^{2}}\mathrm{d} x-L,
\label{constr1}
\end{equation}
with $L>0$ the length of the link and $\phi$ the mode shape function of the AMM \cite{amm} reported in \cite{LAUetal}, that is,
\begin{align}
 \phi(x)=&\cosh\left(\frac{\eta x}{L}\right)-\cos\left(\frac{\eta x}{L}\right)\nonumber \\ &+\gamma\left[\sin\left(\frac{\eta x}{L}\right)-\sinh\left(\frac{\eta x}{L}\right)\right],
 \end{align}
where $\eta$ and $\gamma$ are given in the table \ref{parameters}. The analysis made in \cite{ojas} considers only one mode where the deflection $\alpha(\theta,x)$ is given by
\begin{equation*}
 \alpha(x,\theta)=\phi(x)\theta.
\end{equation*}

The different terms entering into \eqref{EL} are defined as
\begin{align*}
 D(q)&:=\begin{bmatrix} D_1(x_e) && 0 && D_2 (x_e) \\ 0 && D_3 && 0 \\ D_2(x_e) && 0 && D_4 \end{bmatrix},\\ A(q)&=\begin{bmatrix} A_1(\theta,x_e) \\ A_2(\theta,x_e) \\ 0 \end{bmatrix}:=\nabla \Gamma(q)  ,  \;  R:=\diag\{ R_1, 0, R_3 \}, \\  C(q,\dot q)&:=\begin{bmatrix}\hal C_1(x_e) \dot{x}_e && \delta(x_e, \dot \theta, \dot{z}) && \hal C_2(x_e) \dot{x}_e\\ -\delta(x_e, \dot \theta, \dot{z}) &&   0 && -\hal C_2(x_e) \dot \theta \\ \hal C_2(x_e) \dot{x}_e &&  \hal C_2 (x_e)\dot\theta && 0\end{bmatrix},
\end{align*}
with
$$
\delta(x_e, \dot \theta, \dot{z}):=\hal C_1(x_e) \dot \theta+\hal C_2(x_e) \dot{z},
$$
and
\begequ
\lab{nabv}
B(q)=\begin{bmatrix} B_1(\theta,x_e) \\ B_2(\theta,x_e) \\ 0 \end{bmatrix}:=\nabla V(q)
\endequ
where $V$ is the potential energy of the system given by
\begin{equation*}
V(q)=\hal EI \int_0^{x_e}\frac{\left[\theta\phi''(x)\right]^{2}}{\left\lbrace1+\left[\theta\phi'(x)\right]^{2}\right\rbrace^3}\mathrm{d} x-D_3 g(L-x_e),
\end{equation*}
$E,I,D_3,R_1,R_3$ are defined in Table \ref{parameters} and the remaining functions are given in  \ref{appa}. \\

\begin{figure}[!t]
\centering
\includegraphics[scale=.4]{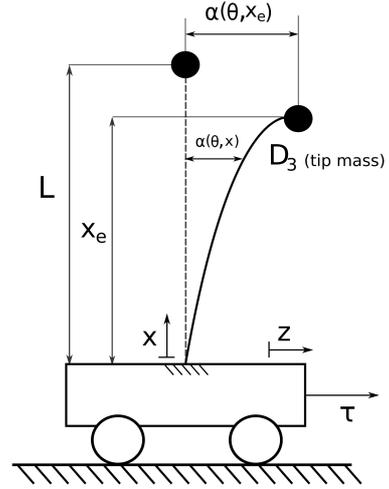}
\caption{Single ultra--flexible link with base excitation}
\label{fig1}
\end{figure}

{\bf Problem formulation:} Given the system (\ref{EL}) find a control input $\tau$ that places the beam at its vertical position with the cart stopped at the zero position, {\em i.e.}, that renders the point $q_*:=(0, L,0)$ a (locally) asymptotically stable equilibrium.\\

\begrem
In \cite{ojas} the model \eqref{EL} is obtained applying EL equations to the constrained Lagrangian
$$
\mathcal{L}(q,\dot q,\lambda) = T(\dot q, q) - V(q) + \lambda \Gamma(q)
$$
where $\lambda$ is a Lagrange multiplier and $T$ is the kinetic energy of the system given by
$$
T(\dot q, q)= \hal \dot q^\top D(q) \dot q.
$$
\endrem

\begrem
It should be noted that the well known \cite{ORTSPO} skew--symmetry property
\begequ
\lab{skesym}
\dot{D}(q)= C(q,\dot q)+C^\top(q,\dot q),
\endequ
is satisfied. Unfortunately, this important property is of no use for controller design in the present context.\\
\endrem

\begrem
In \cite{ojas} the analysis of the open--loop equilibria of \eqref{EL} is carried out. In particular, it is proven that the open--loop equilibrium set is given by
\begin{align}
\nonumber \mathcal{E}:=&\left\lbrace (\theta,x_e,z) \in     \mathbb{D} \times \rea_+ \times \rea \mid\right. \\   &\left. A_1(\theta,x_e)  B_2(\theta,x_e) -A_2(\theta,x_e)  B_1(\theta,x_e) =0\right\rbrace.
\label{Eqmset}
\end{align}
Furthermore, and not surprisingly, it is shown that the desired equilibrium $q_* \in \mathcal{E}$ and is unstable.
\endrem
\subsection{Reduced purely differential model}
\label{subsec22}
%%%%%%%%%%%%%
%
In this subsection we apply the standard constraint differentiation procedure \cite{MARHIL} to transform the algebro--differential equations \eqref{EL} to a purely differential form of reduced order. 
\begin{proposition}\em
\lab{pro1}
The system dynamics \eqref{EL} is equivalent to
\begin{eqnarray}
D_\theta (\theta) \ddot\theta +  D_{z} (\theta) \ddot z + C_\theta  (\theta)\dot\theta^2 + R_1 \dot\theta + B_\theta (\theta) &=& 0 \nonumber \\
D_z (\theta) \ddot \theta + D_4 \ddot z + C_z(\theta)  \dot \theta^2 +R_3 \dot z &=& \tau
\label{srEL}
\end{eqnarray}
with the functions $D_\theta, C_\theta, B_\theta, D_z$ and $C_z$ given in \eqref{coesrEL}.
\end{proposition}
\begin{proof}\em
Differentiating the constraint equation~(\ref{constr1}), we get
\begin{align}
A_{1}(\theta,x_e) \dot{\theta}+A_2(\theta,x_e)  \dot{x}_e&=0 \nonumber\\
A_1(\theta,x_e) \ddot{\theta}+A_2(\theta,x_e) \ddot{x}_{e}&+A_3(\theta,x_e) \dot{\theta}\dot{x}_e\nonumber \\+A_{4}(\theta,x_e)\dot{x}_e^{2}+A_{5}(\theta,x_e) \dot{\theta}^{2}&=0
\label{cons3},
\end{align}
where  the functions $A_i,\;i=3,4, 5,$ are given in the Appendix. Now, the constraint  \eqref{constr1}  satisfies
$$
A_2(\theta,x_e)=\sqrt{1+\left[\theta\phi'(x_e)\right]^2},
$$
which is clearly bounded away from zero. Invoking the \textit{Implicit Function Theorem} \cite{khalil} we can guarantee the existence of a function $\hat x_e(\theta)$ such that
$$
\Gamma(\theta,\hat x_e(\theta))=0.
$$
Equivalently,  it is possible to express $x_e$ in terms of $\theta$, that  is
\begequ
\lab{hatxe}
x_e=\hat x_e(\theta).
\endequ

Replacing \eqref{cons3} in \eqref{EL} it is possible to eliminate the Lagrange multiplier $\lambda$---as done in \cite{ojas}. Moreover, using \eqref{hatxe}, we can  eliminate the coordinate $x_e$ to reduce the order of the system. After some lengthy, but straightforward calculations, this leads to the equations \eqref{srEL} with the definitions
\begin{align}
 D_\theta (\theta)&:= D_1(\hat x_e(\theta))  + D_3 \frac{{A^2_1}(\theta,\hat x_e(\theta))}{A^2_2(\theta,\hat x_e(\theta))} \nonumber \\
 C_\theta (\theta)&:= D_3 \frac{A_1(\theta,\hat x_e(\theta))}{A^2_2(\theta,\hat x_e(\theta))}\zeta
 - \hal C_1 (\hat x_e(\theta))\frac{A_1(\theta,\hat x_e(\theta))}{A_2(\theta,\hat x_e(\theta))}  \nonumber \\
 B_\theta (\theta)&:= B_1(\theta,\hat x_e(\theta))  - B_2(\theta,\hat x_e(\theta))\frac{A_1(\theta,\hat x_e(\theta))}{A_2(\theta,\hat x_e(\theta))} \nonumber \\
D_z (\theta)&:=  D_2 (\hat x_e(\theta)) \nonumber \\
C_z (\theta) &:= -C_2(\hat x_e(\theta))  {A_1(\theta,\hat x_e(\theta))  \over A_2(\theta,\hat x_e(\theta)) } ,
 \label{coesrEL}
\end{align}
where
\begin{align*}
\zeta=&A_5(\theta,\hat x_e(\theta)) + A_4(\theta,\hat x_e(\theta)) \frac{{A^2_1}(\theta,\hat x_e(\theta))}{{A^2_2(\theta,\hat x_e(\theta))}}\\&-A_3(\theta,\hat x_e(\theta))\frac{A_1(\theta,\hat x_e(\theta))}{A_2(\theta,\hat x_e(\theta))}.
\end{align*}
\end{proof}

\begrem
The first equation in \eqref{cons3} can be re-written as follows
\begin{equation}
A^\top(q) \dot q = 0.
\label{PCon}
\end{equation}
Consequently, differentiating the total energy of \eqref{EL}---that is given as $H(q,\dot q):=T(q,\dot q) + V(q)$---and using the skew--symmetry property \eqref{skesym} yields the usual power balance equation
$$
\dot H = - \dot q^\top R \dot q + \dot q_3 \tau.
$$
This means that the virtual forces introduced in the equations due to constrained Lagrange formulation are workless, that is, they are not responsible for addition  or removal of energy from the system. This key property is used later to identify the passive outputs used for the design of the energy shaping controller. 
\endrem

\begrem
The admissible initial conditions for the reduced system \eqref{srEL} are restricted to the set
$$
\{(\theta,z) \in \mathbb{D} \times \rea\;|\; \Gamma(\theta,\hat x_e(\theta))=0\},
$$
where, clearly, the system evolves.
\endrem
%
%%%%%%%%%
\section{Energy Shaping Control}
\label{sec3}
%%%%%%%%%%%%%
%
As explained in the introduction the energy shaping control  of \cite{DONetal} proceeds in three steps: a partial feedback linearization, identification of two passive outputs and the addition of a PID loop around a suitable combination of these outputs. These steps are applied to the system \eqref{srEL} in the following subsections.
\subsection{Partial feedback linearization}
\label{subsec31}
%%%%%%%%%%%%%
%
The lemma below describes a first static state--feedback that performs the PFL of the system  \eqref{srEL}.
%%%
\begin{lemma}\em
\lab{lem1}
Consider the system \eqref{srEL} in closed--loop with the control
\begin{align}
\nonumber \tau =& R_3 \dot z +\left(C_z-\frac{D_z}{D_\theta}C_\theta\right)\dot \theta ^{2}-\frac{D_z}{D_\theta}R_1\dot \theta -\frac{D_z}{D_\theta}B_\theta \\ &+\left(D_4-\frac{D_z^2}{D_\theta}\right)u.
 \label{tau}
\end{align}
Then, the system can be written in Spong's normal form
\begin{eqnarray}
D_\theta (\theta) \ddot\theta  + C_\theta  (\theta)\dot\theta^2 + R_1  \dot\theta + B_\theta (\theta) &=& G_\theta (\theta) u \nonumber \\
 \ddot z& = & u,
\label{rEL}
\end{eqnarray}
where
$$
G_\theta(\theta):=-D_z(\theta).
$$
 \end{lemma}
\begin{proof}\em
The proof proceeds rewriting the first equation of \eqref{srEL} as follows
 \begin{equation}
 \ddot \theta = -\frac{1}{D_\theta}\left(D_z \ddot z+C_\theta\dot \theta ^{2}+R_1\dot \theta +B_\theta\right). \label{ddth}
 \end{equation}
Now, replacing the latter expression in the second equation of \eqref{srEL}, we get

\begin{equation*}
-\frac{D_{z}}{D_\theta}\left(D_z \ddot z+C_\theta\dot \theta ^{2}+R_1\dot \theta +B_\theta\right)  + D_4 \ddot z + C_z  \dot \theta^2 +R_3 \dot z =\tau.
\end{equation*}
Substituting the control law \eqref{tau} in the equation above we obtain the second equality of \eqref{rEL}. The first equation results, replacing $\ddot z=u$ in the first equation of \eqref{srEL}.
\end{proof}
%%%%%%%%

\subsection{Identification of the passive outputs}
\label{subsec32}

In the following lemma the new cyclo--passive\footnote{We recall that the difference between cyclo--passive and  passive maps is that in the former one the storage function is not necessarily bounded from below.} maps for the system \eqref{rEL} are identified. 
%
%%%%%%
\begin{lemma}\em
\lab{lem2}
\label{passive}
Consider the system \eqref{rEL}. The signals
\begequarrs
y_a & := & \dot z\\y_u & := & G_\theta(\theta) \dot \theta,
\endequarrs
define cyclo--passive maps $u \mapsto y_a$ and $u \mapsto y_u$ with storage functions
\begin{eqnarray}
H_a(z)&=&\hal \dot z^2 \label{actH}\\
H_u(\theta)&=& \hal  D_\theta(\theta) \dot \theta^{2}+V_\theta(\theta) \label{underH},
\end{eqnarray}
respectively, where
$$
V_\theta(\theta):=V(\theta,\hat x_e(\theta)).
$$
More precisely, the time derivative of the functions $H_a$ and $H_u$ along the solutions of~(\ref{rEL}) satisfy the dissipation inequalities
\begequ
\dot H _a  \leq  u y_a, \quad
 \dot H _u  \leq  u y_u.
\label{dHuHa}
\endequ
\end{lemma}
\begin{proof}\em
First, notice that
\begequarr
\nonumber\dot V_\theta & = & {\partial V \over \partial \theta}\dot \theta + {\partial V \over \partial \hat x_e}\dot {\hat x}_e\\
\nonumber                      & = & B_1(\theta,x_e) \dot \theta + B_2(\theta,x_e)\dot {\hat x}_e\\
\nonumber                      & = & B_1(\theta,\hat x_e(\theta)) \dot \theta + B_2(\theta,\hat x_e(\theta))\dot x_e\\
\nonumber                      & = &\left[B_1(\theta,\hat x_e(\theta))  - B_2(\theta, \hat x_e(\theta)){A_1(\theta,\hat x_e(\theta)) \over A_2(\theta,\hat x_e(\theta))}\right]\dot \theta   \\
\lab{dotv}                     & = & B_\theta(\theta) \dot \theta,
\endequarr
where we have used \eqref{nabv} for the second identity, \eqref{hatxe} for the third one and the first equation in \eqref{cons3} for the fourth one.

Now we will prove that
\begin{align}
\dot{D}_{\theta}= 2C_\theta \dot\theta.
\label{SkSr}
\end{align}
Indeed, computing the time derivative of $D_{\theta}$ we get
\begin{align*}
 \dot D_{\theta}&=\dot D_1(\hat x_e(\theta))+2D_{3}\frac{A_{1}(\theta,\hat x_e(\theta))}{A_{2}^{2}(\theta,\hat x_e(\theta))}\dot A_{1}(\theta,\hat x_e(\theta)) \\ &-2D_{3}\frac{A_{1}^{2}(\theta,\hat x_e(\theta))}{A_{2}^{3}(\theta,\hat x_e(\theta))}\dot A_2(\theta,\hat x_e(\theta))
 \\&=\frac{\partial D_1}{\partial \hat{x}_e}\dot{\hat{x}}_e+2D_{3}\frac{A_{1}(\theta,\hat x_e(\theta))}{A_{2}^{2}(\theta,\hat x_e(\theta))}\left( \frac{\partial A_{1}}{\partial \theta}\dot{\theta}+ \frac{\partial A_{1}}{\partial \hat{x}_e}\dot{\hat{x}}_e\right)\\ &-2D_{3}\frac{A_{1}^{2}(\theta,\hat x_e(\theta))}{A_{2}^{3}(\theta,\hat x_e(\theta))}\left( \frac{\partial A_{2}}{\partial \theta}\dot{\theta}+ \frac{\partial A_{2}}{\partial \hat{x}_e}\dot{\hat{x}}_e\right)
 \\&=C_1(\hat x_e(\theta)) \dot{x}_e\\&+2D_{3}\frac{A_{1}(\theta,\hat x_e(\theta))}{A_{2}^{2}(\theta,\hat x_e(\theta))}\left[ A_5(\theta,\hat x_e(\theta)) \dot{\theta}+\frac{1}{2}A_3(\theta,\hat x_e(\theta))\dot{x}_e \right. \\ &\left. -\frac{A_{1}(\theta,\hat x_e(\theta))}{A_{2}(\theta,\hat x_e(\theta))}\left( \frac{1}{2}A_3(\theta,\hat x_e(\theta))\dot{\theta}+A_4(\theta,\hat x_e(\theta))\dot{x}_e\right)\right]
 \\&=\left[-C_{1}(\hat x_e(\theta))\frac{A_{1}(\theta,\hat x_e(\theta))}{A_{2}(\theta,\hat x_e(\theta))}+2D_{3}\frac{A_{1}(\theta,\hat x_e(\theta))}{A_{2}^{2}(\theta,\hat x_e(\theta))}\zeta\right]\dot \theta
 \\&=2C_\theta \dot\theta,
\end{align*}
where we used the first equation of \eqref{cons3} to eliminate $\dot x_e$ in the fourth identity.

The time derivative of \eqref{underH} along the system trajectories is
\begin{align*}
\dot H _u&=  D_\theta \dot\theta  \ddot \theta + \hal \dot D_\theta \dot \theta^2+\dot V_\theta \\
&=\left( G_\theta  u-C_\theta\dot \theta ^{2}-R_1\dot \theta +\hal \dot D_\theta\dot \theta\right)\dot \theta \\
&=-R_1 \dot\theta^{2}+G_\theta \dot \theta u\leq  u y_u.
\end{align*}
where \eqref{ddth} and \eqref{dotv} were used in the second equality while the third one was obtained invoking \eqref{SkSr}.

On the other hand, the time derivative of \eqref{actH} along the system trajectories is
\begin{align*}
\dot H _a&=\dot z \ddot z = u y_a.
\end{align*}
This completes the proof.
\end{proof}
\subsection{PID controller}
\label{subsec33}
%%%%%%%%%%
%
Similarly to \cite{DONetal} the controller design is completed adding a PID around a suitably weighted sum of the two cyclo--passive outputs ($y_a$ and $y_u$) identified in Lemma \ref{lem2}. More precisely, the controller implements the relationship
\begequ
\lab{pidcon}
k_e u = -\left(K_P \tilde y+{K_I }\int_0^t \tilde y(s)ds + K_D \dot{\tilde y}\right),
\endequ
where
\begequ
\lab{tily}
\tilde y:=k_a y_a + k_u y_u
\endequ
with $k_e,k_a,k_u \in \rea$ and $K_P,K_I,K_D \in \rea_+$  the PID gains. As explained in \cite{DONetal}, and illustrated below, these gains are selected to shape the energy function.

To implement the controller \eqref{pidcon} {\em without differentiation} the term   $ \dot{\tilde y}$ is replaced by its evaluation along the system dynamics \eqref{rEL}. Since the system is relative degree one this brings along some terms depending on $u$ that are moved to the left hand side of \eqref{pidcon}. Some lengthy, but straightforward, calculations show that \eqref{pidcon} is equivalent to
$$
K(\theta) u =   -\left(K_P \tilde y+{K_I }\int_0^t \tilde y(s)ds\right) -K_D k_u S(\theta,\dot \theta),
$$
where we defined the functions
 \begin{align}
S(\theta,\dot \theta)&:=\dot G_\theta \dot\theta - \frac{G_\theta}{D_\theta}\left[C_\theta {\dot \theta}^2+R_1 \dot \theta + B_\theta\right]
\nonumber\\
K(\theta)&:=k_{e}+K_D\left[k_{a}+k_{u}\frac{G_{\theta}^{2}(\theta)}{D_\theta(\theta)}\right].
\nonumber
\end{align}
Clearly, a sufficient condition for the controller to be implementable is that the function $K$ is bounded away from zero, that is,
\begequ
|K| \geq \delta >0.
\label{cond3}
\endequ

To analyse the stability of the system \eqref{rEL} in closed--loop with the PID \eqref{pidcon}, \eqref{tily} we propose the function
\begin{align*}
W(t,\tilde y, \theta,\dot \theta,\dot z):=&   k_e[k_a H_a(\dot z)+k_u H_u(\theta,\dot \theta)]\\&+ {K_I \over 2}\left(\int_0^t \tilde y(s)ds\right)^2+ {K_D \over 2} \tilde y^2,
\end{align*}
and make the reasonable assumption that the friction forces acting on the beam are negligible, hence set $R_1=0$. The derivative of $W$  yields
\begequarrs
\dot W & = & k_e(k_a \dot H_a+k_u \dot H_u)+{K_I}\tilde y \int_0^t \tilde y(s)ds+ {K_D \over 2} \tilde y \dot {\tilde y}\\
          & = & k_e \tilde y u +{K_I}\tilde y \int_0^t \tilde y(s)ds+ {K_D \over 2} \tilde y \dot {\tilde y}\\
          & = & -K_P \tilde y^2,
\endequarrs
where we used the dissipation inequalities \eqref{dHuHa}---that under the assumption $R_1=0$ become equalities---to get the second identity and replaced \eqref{pidcon} to find the last one.

The final step in our stability analysis is to show that the function $W$ can be expressed as a positive definite (with respect to the desired equilibrium) function of the state $(\theta,z, \dot \theta,\dot z)$ of the system \eqref{rEL}. Notice that for this reduced system the desired equilibrium is simply the {\em origin}.

To express $W$  as a function of the state we only need to deal with the integral term. For, we define the function
$$
V_{N}(\theta):=-D_3\int^{\theta}_{0}\phi(\hat{x}_e(s))\mathrm{d} s-\left(\rho \mathcal{A}_0 \int^L_0 \phi(x) \mathrm{d}x\right)\theta,
$$
whose time derivative is given by
 \begin{align}
  \nonumber \dot{V}_N&= -D_3\phi(\hat x_e (\theta))\dot{\theta}-\left(\rho \mathcal{A}_0 \int^L_0 \phi(x) \mathrm{d}x\right)\dot{\theta} \\
  \nonumber &=-D_{z}\dot\theta =G_\theta\dot{\theta} =y_u.
  \label{VNdot}
 \end{align}
Consequently
$$
\int_0^t \tilde y(s)ds=k_a z(t)+ k_u V_N(\theta(t)) + c,
$$
where $c \in \rea$ is an integration constant. Using the latter and the definitions of $H_a,H_u$ and $\tilde y$ we can prove that, up to an additive constant,

\small
\begin{equation}
W(t,\tilde y, \theta,\dot \theta,\dot z) =\hal \begin{bmatrix} \dot \theta \\ \dot z  \end{bmatrix}^\top D_d(\theta) \begin{bmatrix} \dot \theta \\ \dot z  \end{bmatrix} +V_d(\theta,z)=:H_d(\theta,z,\dot \theta,\dot z)
\label{HD}
\end{equation}
\normalsize
where we defined
\begin{align}
 D_{d}(\theta)&:=\begin{bmatrix}
             k_{e}k_{u}D_\theta(\theta) +k_{u}^{2}K_DG_\theta^2(\theta) && k_{a}k_{u}K_DG_\theta(\theta)\\  k_{a}k_{u}K_D G_\theta(\theta) && k_{e}k_{a}+k_{a}^{2}K_D
              \end{bmatrix}
              \label{Dd}\\
V_d(\theta,z)&:= k_{e}k_{u}V_\theta(\theta)+\hal K_{I}\left[k_{a}z+k_{u}V_{N}(\theta)\right]^{2}.
\label{Vd}
\end{align}

\begrem
Without the assumption that $R_1=0$ a term $-k_ek_u R_1 \dot \theta^2$ appears in $\dot W$. As will be shown below, see also \cite{DONetal} and  Remark \ref{rem7}, to make the upward position a minimum of the total energy function $H_d$ it is necessary to flip the potential energy of the pendulum, which is done selecting $k_ek_u < 0$, making positive the dissipation term. The deleterious effect of dissipation in energy shaping methods is well known and has been reported in various references \cite{WOOLetal,GOMAvdS}.
\endrem
\subsection{Main stability result}
\label{subsec34}
%%%%%%%%%%
%
The proposition below, which essentially  gives conditions on the controller gains to ensure $H_d$ is positive definite, is the main result of this section. 
%
%%%
\begin{proposition}\em
\lab{pro2}
Consider the system \eqref{EL} in closed--loop with the controller \eqref{tau} where the outer--loop control $u$ is given by the PID \eqref{pidcon} with
\begequ
\lab{tiltily}
\tilde y = k_a \dot z + k_u G_\theta(\theta) \dot \theta.
\endequ

Set the constant gains $k_e,k_a$ and the PID gains $K_P,K_I$ and $K_D$ to {\em arbitrary positive} numbers while $k_u$ is negative and, for some small $\epsilon >0$, satisfies
\begequ
\lab{conku}
k_{u} \leq  - \calc\left(k_a+ {k_e \over K_D}\right)-\epsilon,
\endequ
where $\calc$ is a positive constant verifying
\begequ
\lab{M}
\calc \geq \frac{D_\theta(\theta)}{G^2_\theta(\theta)}.
\endequ
\begite
\item[(i)] The origin of the reduced dynamics \eqref{srEL}, which corresponds to the desired equilibrium $q_*=(0,L,0)$ of \eqref{EL}, is {\em stable} with Lyapunov function $H_d$  given in \eqref{HD}.
\item[(ii)] It is {\em asymptotically} stable if the signal $\tilde y$ defined in \eqref{tiltily} is detectable with respect to \eqref{rEL}.\\
\endite

\end{proposition}
%
%%%%%%
\begin{proof}\em
In Lemma \ref{lem1} it has been shown that the dynamics of the system \eqref{EL} in closed--loop with the controller \eqref{tau} is described by \eqref{rEL}. Therefore, given the derivations above, it only remains to prove that the non--increasing function $H_d$, defined in \eqref{HD}, is positive definite. This will be established proving that, under the conditions of the proposition $D_d >0$ and $V_d$ has an isolated minimum at the origin.

To prove the first claim notice from \eqref{Dd} that the $(2,2)$ term of $D_d$ is positive. Hence it only remains to show that its determinant is also positive. Now,
\begin{align*}
 \det \{D_d\}&=k_{e}k_{u}D_\theta\left( k_{e}k_{a}+k_{a}^{2}K_D \right)+k_{e}k_{u}^{2}k_{a}K_DG_\theta^{2} \\
 &=k_{e}k_{u}\left[D_\theta\left(  k_{e}k_{a}+k_{a}^{2}K_D \right)+k_uk_aK_DG_\theta^{2}\right]\\
 &=k_{e}k_{u}k_{a}K_D{G^2_\theta}\left[\frac{D_\theta}{G^2_\theta}(k_{a} +{k_{e} \over K_{D}}) +k_u\right].
\end{align*}
Since $k_ek_u < 0$, the term outside the brackets is negative. Furthermore, if \eqref{conku} and \eqref{M} are satisfied, the term inside the brackets is also negative, yielding $ \det \{D_d\} > 0$.

We proceed now to prove that the condition \eqref{cond3}, which ensures realisability of the control \eqref{pidcon}, is satisfied. This is established noticing that
$$
\frac{D_\theta(\theta)}{G^2_\theta(\theta)} K_D K(\theta) =\frac{D_\theta(\theta)}{G^2_\theta(\theta)} \left(k_a+ {k_e \over K_D}\right) + k_{u}.
$$
Hence, invoking \eqref{conku}, we have that  \eqref{cond3} holds.

To establish the proof of the second claim we compute the gradient of $V_d$ as
\begin{align*}
\nabla V_{d}&= \begin{bmatrix}
                                       k_ek_u \nabla V_\theta +K_Ik_u \nabla V_N \left(k_{a}z+k_{u}V_{N}\right)   \\K_I k_a\left( k_{a}z+k_{u}V_{N}\right)
                                         \end{bmatrix} \\
                                         &=\begin{bmatrix}
                                                                                 k_ek_uB_\theta-K_Ik_u D_z \left(k_{a}z+k_{u}V_{N}\right)   \\K_I k_a\left( k_{a}z+k_{u}V_{N}\right)
                                         \end{bmatrix}.
\end{align*}
Using the fact that $B_\theta(0)=0$ and $V_N(0)=0$ we conclude that $\nabla V_d(0)=\col(0,0)$.
Now, the Hessian of $V_d$ is given by
\small
 \begin{align*}\nonumber
 \nabla^{2} V_{d}&= \nonumber\begin{bmatrix}
                     k_ek_u \nabla^{2}V_\theta + K_Ik_u\nabla ^{2} ( k_{a}z+k_{u}V_{N})
                                       & -K_Ik_uk_aD_z\\-K_Ik_uk_aD_z & K_Ik_a^{2}
                    \end{bmatrix}\\ \nonumber\\
&=\begin{bmatrix}
    \nu(\theta,z)
                                       & -K_Ik_uk_aD_z\\-K_Ik_uk_aD_z & K_Ik_a^{2}
  \end{bmatrix},
\end{align*}
\normalsize
where we defined the function
$$
\nu(\theta,z):=k_ek_u \nabla^{2}V_\theta+K_Ik_u^2 D_z^2-K_Ik_u\nabla D_z\left( k_{a}z+k_{u}V_{N}\right).
$$
Evaluating it at the origin yields
\begin{align}
\nabla^{2} V_{d}(0)=\begin{bmatrix}
                                          \nu(0) & -K_Ik_uk_aD_z(0) \\
                                         -K_Ik_uk_aD_z(0) &  K_Ik_a^{2}
                                         \end{bmatrix}, \label{evalhess}
\end{align}
where
$$
\nu(0)=k_ek_u \nabla^2 V_\theta(0)+ K_Ik_u^2 D_z^2(0).
$$
\vspace{-.2cm}
Now,
\vspace{-.2cm}
\begin{align*}
\nabla^2 V_\theta(0)=EI\int^{\hat{x}_e(0)}_{0}[\phi '' (x)]^{2}\mathrm{d}x -D_3g\int^{\hat{x}_e(0)}_{0}[\phi ' (x)]^{2}\mathrm{d}x,
\end{align*}
which can be shown to be negative \cite{ojas}. Since $k_ek_u <0$ the $(1,1)$ term of $\nabla^{2} V_{d}(0)$ is positive. Moreover,
$$
\det \{  \nabla^{2} V_{d}(0)\}= k_ek_u \nabla^2 V_\theta(0)  K_Ik_a^{2},
$$
which is also positive, ensuring$\nabla^{2} V_{d}(0)>0$.

The previous analysis ensures that the origin is an isolated minimum for the function $V_d$ as claimed above. The proof is completed invoking classical Lyapunov theory \cite{khalil}.
\end{proof}

\begrem
\lab{rem7}
Notice that the condition $\nabla^2 V_\theta(0)<0$ is consistent with the well known fact that the upward pendulum position is unstable in open--loop. Similarly to the rigid case \cite{DONetal} the maximum of the open--loop potential energy is transformed into a minimum in closed--loop multiplying $V_\theta$ by the negative number $k_ek_u$---see \eqref{Vd}. 
\endrem

\begrem
The critical condition \eqref{conku} is satisfied in a neighborhood of the origin replacing $C$ by
$$
\frac{D_3\phi^2(L)+\rho \mathcal{A}_0 \int_0^L \phi^2(x)dx}{\left[D_3\phi(L)+\rho \mathcal{A}_0 \int_0^L \phi(x)dx\right]^2}=\frac{D_\theta(0)}{G^2_\theta(0)}.
$$
\endrem
\vspace{.1cm}
\begrem
The term $k_{a}z+k_{u}V_{N}(\theta)$ in  \eqref{Vd} is a new potential energy corresponding to a virtual spring attached to the cart---thereby enabling to stabilize the cart position. 
\endrem

\begrem
The choice of the free gains of Proposition \ref{pro2} is given just as an illustration. From the proof it is clear that, depending on the particular problem, other (possibly less conservative) choices are available.
\endrem
%
%%%%%%%%%

%\vspace{-.3cm}
\section{Simulation Results}
\label{sec4}
%%%%%%%%%%%%%%%%%%
%\vspace{-.1cm}

In this section we assess the performance of the proposed controller via Matlab simulations choosing different sets of gains and different initial conditions. We simulated the system \eqref{rEL} in closed--loop with the PID controller \eqref{pidcon}, \eqref{tiltily} with the parameters given in Table~\ref{parameters}. 

We have chosen three different initial conditions, given in Table \ref{ICs}, corresponding to radically different scenarios of the system. Namely, an arbitrary point (ICs1), one of the stable open-loop equilibria (ICs2)
and an initial condition with the cart far from the origin and the tip mass located at the unstable open--loop equilibrium (ICs3).

For the selection of suitable gains for the controller, we fixed the gain $k_e=1$ and linearized the closed--loop system. We based our criterion to choose the gains, always satisfying \eqref{cond3} and \eqref{conku},
and observing the eigenvalues of the closed--loop matrix of the linearized system around the desired equilibrium point. Particular attention has been paid to the eigenvalue closest to the imaginary axis, which is directly related to the rate of convergence of the cart position. Three sets of gains were selected and they are given in Table \ref{Gs}. For the Set 1 the real part of the slowest pole was $-0.58$, $-0.75$ for the Set 2 and $-1.33$ for Set 3.

Simulation results of the energy shaping control are shown in Figures~\ref{sIC1}--\ref{sIC3}, where the variation of the cart position and control input acceleration is
observed to be within practical limits, hence the control objective of simultaneous stabilisation of cart position while suppressing the cantilever vibrations is achieved. 

To evaluate the effect of the gains on the estimate of the domain of attraction of the closed--loop systems provided by the Lyapunov function $H_d$ we show in Figure \ref{lc} the level curves of the desired potential energy $V_d$ for each set of gains. As expected, there is a tradeoff between convergence rate and the size of the domain of attraction---as the slowest closed--loop pole of the linearized system moves farther to the left the closed sublevel sets shrink.
%
%
%%%%%%%%%%%%%%%%%%
%
%%%%%%%%%

%\vspace{-.1cm}
\section{Experimental Results}
%%%%%%%%%%%%%%%%%%
\label{secexp}
%\vspace{-.1cm}

Experiments were also carried out to assess the performance of the proposed controller. The physical description of the experimental setup is provided in \ref{appb}. Although the simulations of the previous section demonstrate effectiveness with three sets of gains, it was found that the control becomes unstable for the same gains in the experiments and it was necessary to retune them via a trail-and-error procedure. We believe this problem was due to a discrepancy in the values of the parameters of the model, which were identified 
in~\cite{ojas}  to capture the multiple equilibria and chaotic behaviour observed experimentally. A different parameter identification process is necessary to reproduce the smooth dynamical transients for small beam deviations, which is the behaviour assumed for the controller design. 

The partial feedback linearization was replaced by the standard procedure of obtaining the desired trajectories for $z$ via the integration of the cart acceleration, which is numerically reconstructed. 

The set of gains used in the experiment are $k_e=1$, $k_a=1$, $k_u=-47.5$, $K_D=1.9$, $K_P=3$ and $K_I=0.9$.  Simulation results are repeated for this set of gains. Figure~\ref{sx} presents the comparison of simulation results with experimental counterparts for the set of initial conditions ICs 2. The results demonstrate that the control task is achieved in a similar time although the trajectory in the experiments shows more oscillations with high frequency components of the vibrations of the beam. These oscilations are not captured by the simulation model that, as explained in Section \ref{sec2}, retains only the first deflection mode.  However, as shown in the plots, these high frequency vibrations degrade the transient performance but do not induce instability.

A video of the experiments  can be watched in https://youtu.be/YBcI9WIaQa0.
%

%%%%%%%%%
%\vspace{-.1cm}
\section{Conclusions and Future Research}
\label{sec5}
%%%%%%%%%%%%%%%%%%
%
%\vspace{-.1cm}
The new energy shaping, fully constructive, controller of \cite{DONetal} has been applied for the stabilisation of the upward unstable position of a cart with an ultra--flexible beam. The dynamics of this system is accurately described by an EL system with an algebraic constraint. The reduction of this dynamics to the constrained manifold is not an EL system, therefore the technique of \cite{DONetal} must be modified for its application in our system. It is shown in the paper that, due to the lossless nature of the constraint forces, it is still possible to identify the passive outputs required for the controller design, which is completed adding a classical PID controller around a suitable combination of these outpus. 

The performance of the proposed controller is assessed via simulations and experiments. The experiments showed that the parameters identified in \cite{ojas} do not capture the dynamic behaviour required for the controller design. A new identification stage, aimed at correcting this deficiency, is currently under way. Another line of research that we are pursuing pertains to the development of a theory for general constrained EL systems that will extend the work done in  \cite{DONetal} for unconstrained EL systems.

%%%
%\vspace{-.1cm}
\section*{Acknowledgment}
%\vspace{-.1cm}
Pablo Borja would like to thank the mexican Secretary of Public Education (SEP) and the mexican National Council of Science and Technology (CONACyT) for supporting his research.

The work of Romeo Ortega was supported in part by the Government of Russian Federation
(Grant 074-U01) and the Ministry of Education and Science of Russian
Federation (Project 14.Z50.31.0031).

\newpage

\onecolumn

\begin{table}[!t]
\caption{System parameters}
%\vspace{-.4cm}
\large
\label{parameters}
\begin{center}
\begin{tabular}{|l|c|c|c|} \hline 
\multicolumn{1}{|c|}{\textbf{Parameter}}
 & \textbf{Symbol} & \textbf{Value} & \textbf{Units} \\\hline
Pendulum cross section area & $\mathcal{A}_o$ & $8 \times 10^{-6}$ & $m^2$\\\hline
Young's modulus & $E$ & $9\times 10^{10}$ & $\frac{N}{m^2}$\\\hline
Gravitational acceleration & $g$ & $9.81$ & $\frac{m}{seg^2}$\\\hline
Moment of inertia & $I$ & $1.066 \times 10^{-13}$ & $kg\cdot m^2$\\\hline
Pendulum length & $L$ & $0.305$ & $m$\\\hline
Tip mass & $M$ & $2.75\times 10^{-2}$ & $kg$\\\hline
Cart mass & $M_c$ & $0.1$ & $kg$\\\hline
Function of the system natural frequency & $\eta$ & $1.1741$ & $-$\\\hline
Dimensionless constant & $\gamma$ & $0.9049$ & $-$\\\hline
Pendulum density & $\rho$ & $8400$ & $\frac{kg}{m^3}$\\\hline
Viscous friction at the pendulum base & $R_1$ & $9.86\times 10^{-4}$ & $\frac{kg}{seg}$\\\hline 
Viscous friction between the rail and the cart & $R_3$ & $7.69$ & $\frac{kg}{seg}$ \\\hline
\end{tabular}
\end{center}
\end{table}

%\twocolumn

%\onecolumn

\begin{table}[!h]
\parbox{10cm}{
\caption{Initial conditions}
\vspace{.28cm}
\hspace{1cm}
\label{ICs}
\begin{tabular}{c|c|c|c|c|}
× & $\theta \ [m]$ & $z \ [m]$ & $\dot{\theta} \ [m/s]$ & $\dot{z} \ [m/s]$ \\\hline
\multicolumn{1}{|c|}{ ICs 1}  & $-0.08$ & $-0.1$ & $0$ & $0$\\\hline
\multicolumn{1}{|c|}{ ICs 2} & $0.134$ & $0 $ & $0$ & $0$\\\hline
\multicolumn{1}{|c|}{ ICs 3} & $0$ & $-0.15$ & $0$ & $0$\\\hline
\end{tabular}
}
\hspace{-2cm}
\parbox{10cm}{
\centering
\caption{Gains sets}
\vspace{.2cm}
\label{Gs}
\begin{tabular}{c|c|c|c|c|c|c|}
× & $\ \ k_e$ \ \ &\ \ $k_a$ \ \ & $k_u$ &\ \ $K_D$ \ \ & \ \ $K_P$ \ \ & \ \ $K_I$ \ \ \\\hline
\multicolumn{1}{|c|}{ Set 1} & $1$ & $0.5$ & $-50.77$ & $1.47$ & $1.94$ & $0.35$ \\\hline
\multicolumn{1}{|c|}{ Set 2} & $1$ & $1$ & $-61.37$ & $1.28$ & $1.92$ & $0.52$ \\\hline
\multicolumn{1}{|c|}{ Set 3} & $1$ & 1 & $-43.04$ & $2.18$ & $3.66$ & $1.35$\\\hline
\end{tabular}
%\caption{Bar}
}
\end{table}

% \begin{table}[!h]
%{\caption{Initial conditions}
%\vspace{-.3cm}
%\label{ICs}
 %\begin{center}
 %\caption{Initial conditions}
%\vspace{-.3cm}
%\small
%\begin{tabular}{c|c|c|c|c|}
%× & $\theta \ [m]$ & $z \ [m]$ & $\dot{\theta} \ [m/s]$ & $\dot{z} \ [m/s]$ \\\hline
%ICs 1 & $-0.08$ & $-0.1$ & $0$ & $0$\\\hline
%ICs 2 & $0.134$ & $0 $ & $0$ & $0$\\\hline
%ICs 3 & $0$ & $-0.15$ & $0$ & $0$\\\hline
%\end{tabular}
%\end{center}
 %} %
 % \end{table}
% \begin{table}[!h]
%{\caption{Gains sets}
%\label{Gs}
%\vspace{-.3cm}
%\scriptsize
% \begin{center}
%\renewcommand{\arraystretch}{1.3}
%\begin{tabular}{c|c|c|c|c|c|c|}
%× & $\ \ k_e$ \ \ &\ \ $k_a$ \ \ & $k_u$ &\ \ $K_D$ \ \ & \ \ $K_P$ \ \ & \ \ $K_I$ \ \ \\\hline
%Set 1 & $1$ & $0.5$ & $-50.77$ & $1.47$ & $1.94$ & $0.35$ \\\hline
%Set 2 & $1$ & $1$ & $-61.37$ & $1.28$ & $1.92$ & $0.52$ \\\hline
%Set 3 & $1$ & 1 & $-43.04$ & $2.18$ & $3.66$ & $1.35$\\\hline
%\end{tabular}
%\end{center}
%\normalsize
%}
%\end{table}

 \begin{figure}[!h]
 \hspace{-.23cm}
\includegraphics[scale=.4]{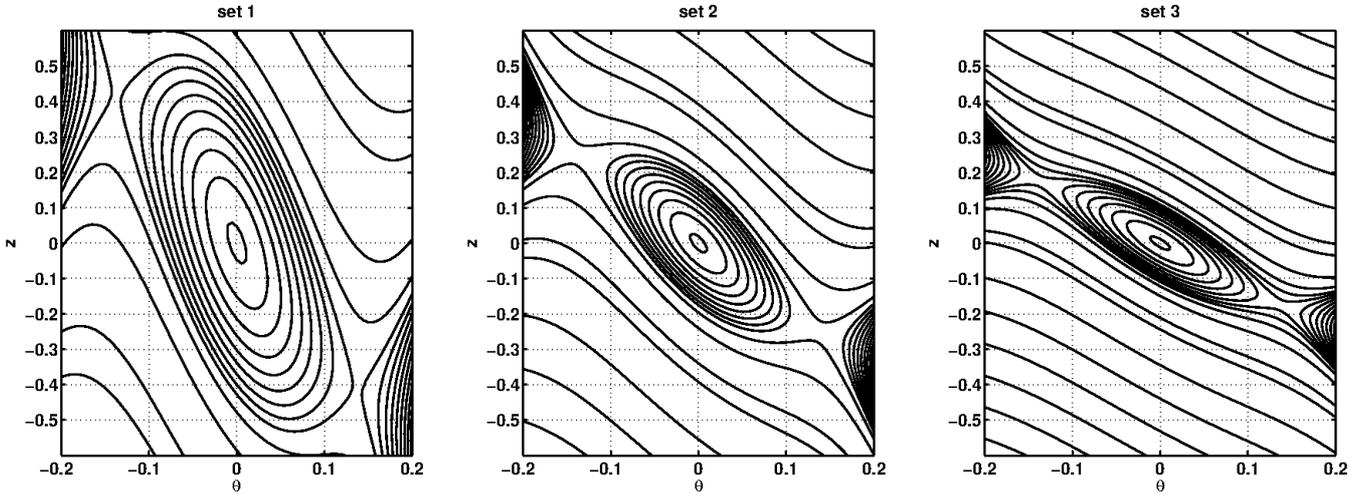}
\caption{Level curves of the desired potential energy $V_d(\theta,z)$ for the different sets of gains}
\label{lc}
%\end{minipage}
\end{figure}

\newpage

\begin{figure}[!t]
 \onecolumn
 \hspace{-2cm}
 \includegraphics[scale=.4]{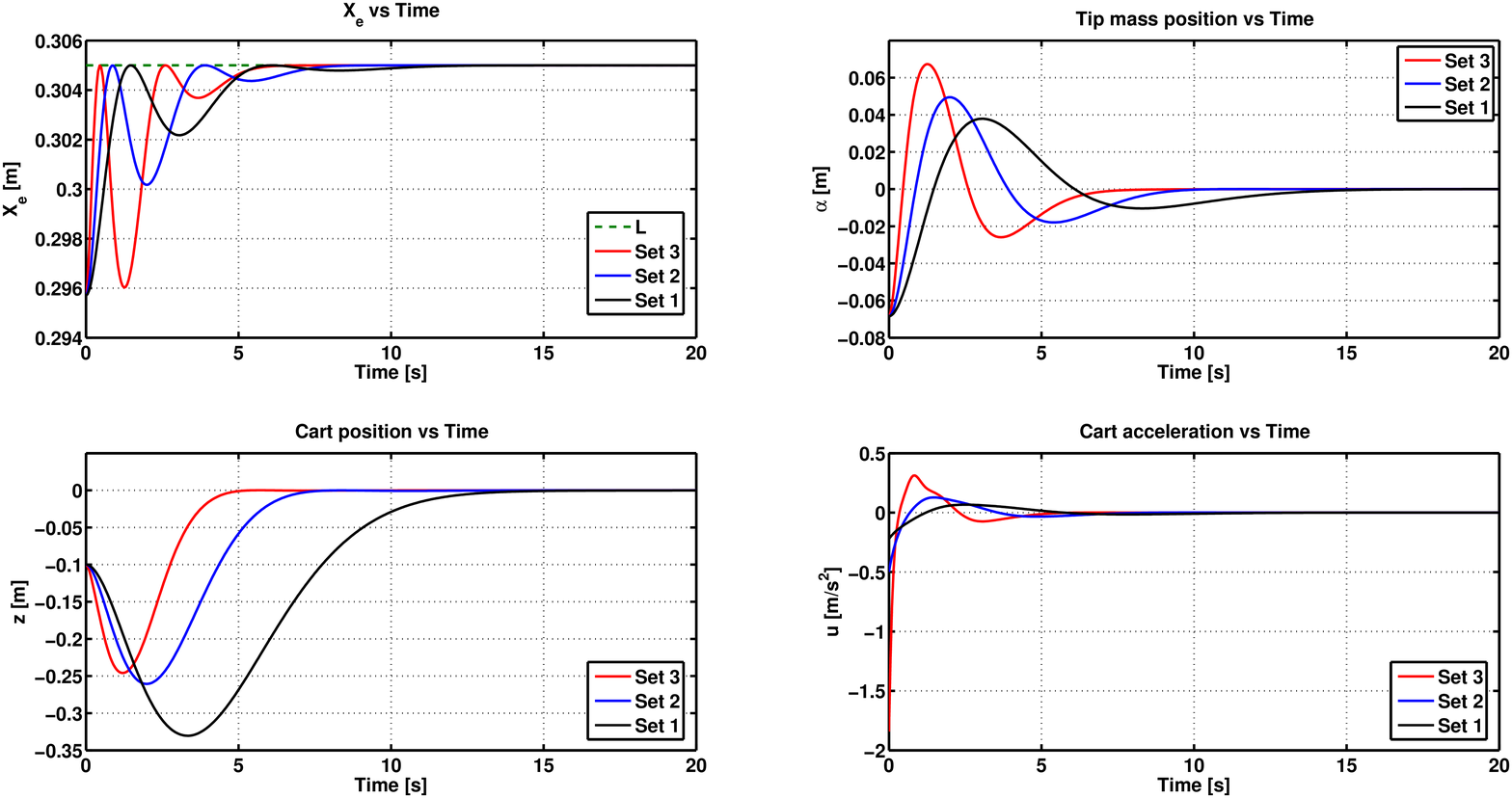}
 % edosic1.eps: 0x0 pixel, 300dpi, 0.00x0.00 cm, bb= -459    30  1072   762
\caption{ Simulation results for ICs 1}
 \label{sIC1}
%\end{figure}
%\twocolumn
%\begin{figure}
%\onecolumn
 \hspace{-2cm}
 \includegraphics[scale=.4]{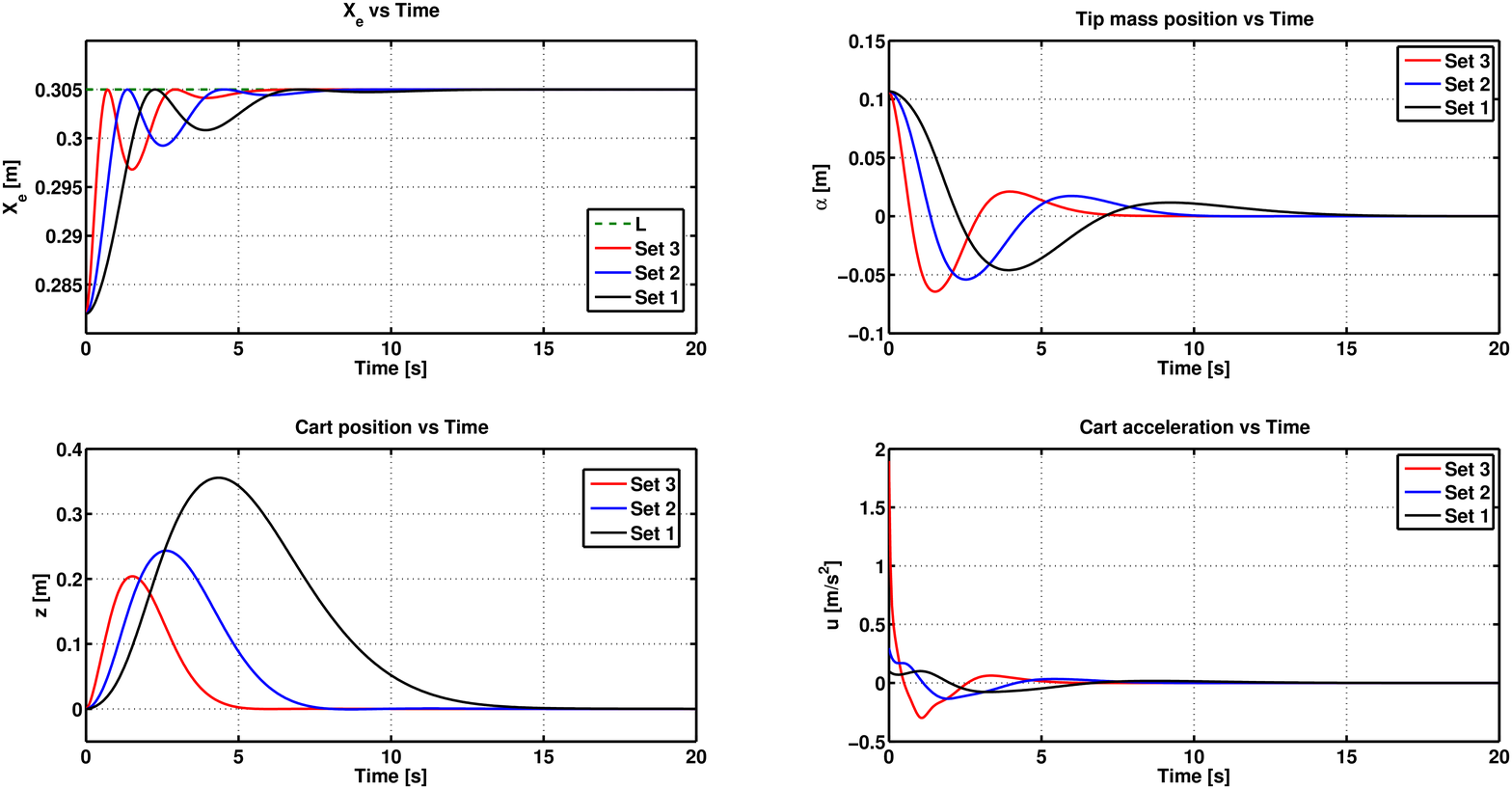}
 % edosic2.eps: 0x0 pixel, 300dpi, 0.00x0.00 cm, bb= -459    33  1072   757
 \caption{Simulation results for ICs 2}
 \label{sIC2}
\end{figure}
\twocolumn

\newpage
\onecolumn
 \begin{figure}[!t]
 %\begin{minipage}[!t]{\textwidth}
  \hspace{-2cm}
 \includegraphics[scale=.4]{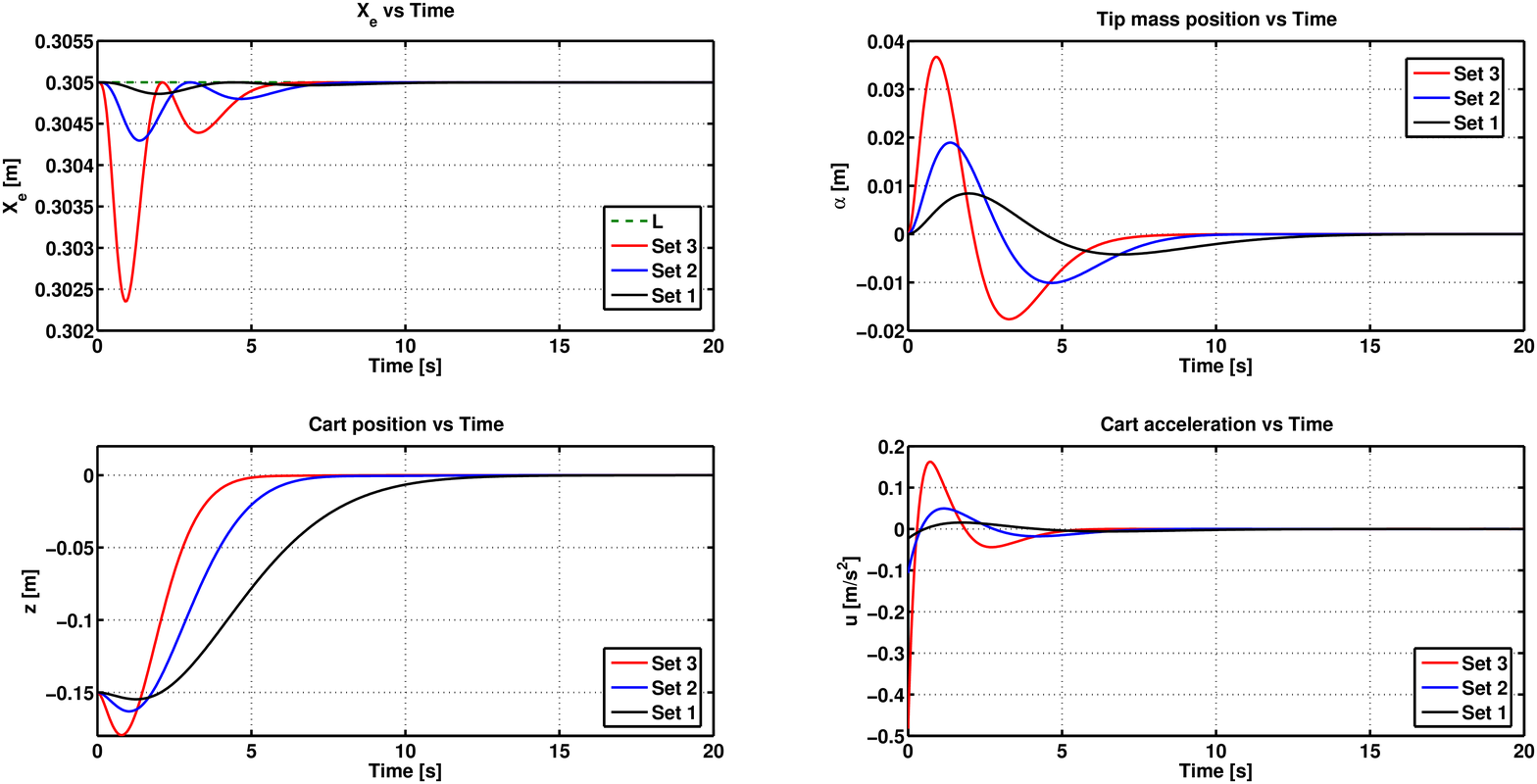}
 % edosic3.eps: 0x0 pixel, 300dpi, 0.00x0.00 cm, bb= -459    40  1072   752
 \caption{Simulation results for ICs 3}
 \label{sIC3}
 \end{figure}

  \begin{figure}[!t]
  \onecolumn
\hspace{-1.7cm}
 \includegraphics[scale=.36]{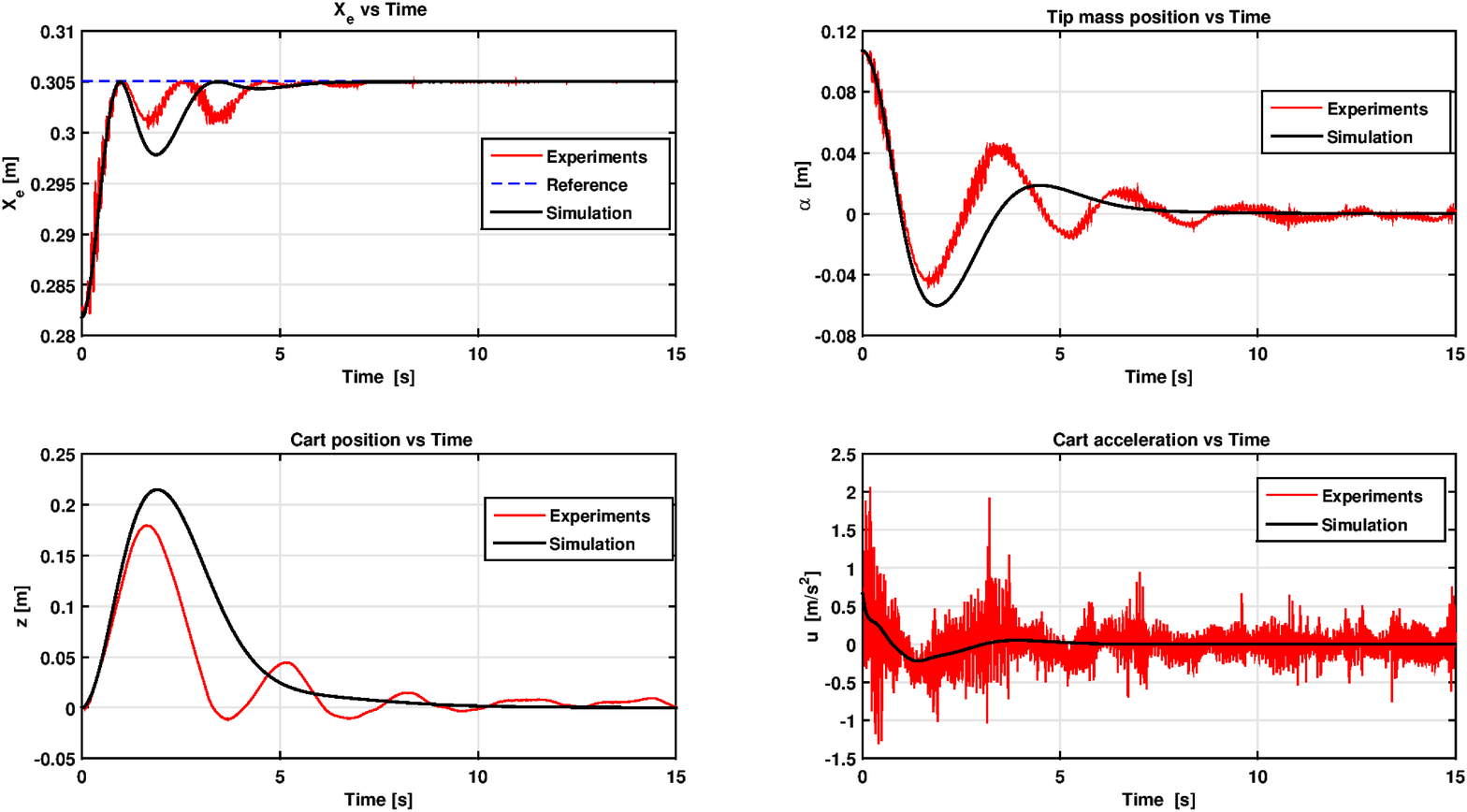}
  %edosic3.eps: 0x0 pixel, 300dpi, 0.00x0.00 cm, bb= -459    40  1072   752
 \caption{Comparison of simulation and experimental results for ICs 2}  
 \label{sx}
\end{figure}

 \twocolumn
%%%%%%%%%%%%%%

\begin{figure}[!t]
  \centering
 %\hspace{-2cm}
 \includegraphics[width=.55\linewidth]{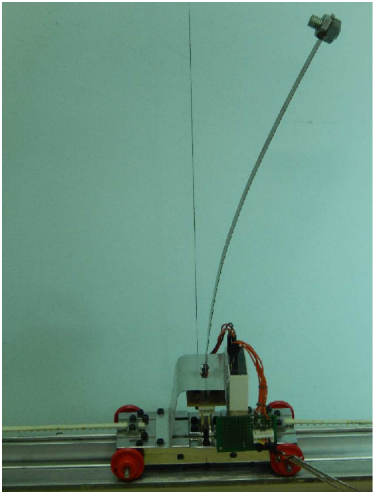}
  %edosic3.eps: 0x0 pixel, 300dpi, 0.00x0.00 cm, bb= -459    40  1072   752
 \caption{Inverted flexible pendulum.}  
 \label{pendulum}
\end{figure}

\bibliographystyle{elsarticle-num-names}
\bibliography{biblio}

\begin{thebibliography}{14}
\providecommand{\natexlab}[1]{#1}
\providecommand{\url}[1]{\texttt{#1}}
\providecommand{\urlprefix}{URL }
\expandafter\ifx\csname urlstyle\endcsname\relax
  \providecommand{\doi}[1]{doi:\discretionary{}{}{}#1}\else
  \providecommand{\doi}[1]{doi:\discretionary{}{}{}\begingroup
  \urlstyle{rm}\url{#1}\endgroup}\fi
\providecommand{\bibinfo}[2]{#2}

\bibitem[{Meirovitch(1975)}]{amm}
\bibinfo{author}{L.~Meirovitch}, \bibinfo{title}{Elements of vibration
  analysis}, \bibinfo{publisher}{McGraw-Hill}, \bibinfo{year}{1975}.

\bibitem[{Dwivedy and Eberhard(2006)}]{dsurvey}
\bibinfo{author}{S.~K. Dwivedy}, \bibinfo{author}{P.~Eberhard},
  \bibinfo{title}{Dynamic analysis of flexible manipulators, a literature
  review}, \bibinfo{journal}{Mechanism and machine theory}
  \bibinfo{volume}{41}~(\bibinfo{number}{7}) (\bibinfo{year}{2006})
  \bibinfo{pages}{749--777}.

\bibitem[{Patil and Gandhi(2014)}]{ojas}
\bibinfo{author}{O.~Patil}, \bibinfo{author}{P.~Gandhi}, \bibinfo{title}{On the
  dynamics and multiple equilibria of an inverted flexible pendulum with tip
  mass on a cart}, \bibinfo{journal}{Journal of Dynamic Systems, Measurement,
  and Control} \bibinfo{volume}{136}~(\bibinfo{number}{4})
  (\bibinfo{year}{2014}) \bibinfo{pages}{041017}.

\bibitem[{Torfs et~al.(1998)Torfs, Vuerinckx, Swevers, and Schoukens}]{leadlag}
\bibinfo{author}{D.~E. Torfs}, \bibinfo{author}{R.~Vuerinckx},
  \bibinfo{author}{J.~Swevers}, \bibinfo{author}{J.~Schoukens},
  \bibinfo{title}{Comparison of two feedforward design methods aiming at
  accurate trajectory tracking of the end point of a flexible robot arm},
  \bibinfo{journal}{Control Systems Technology, IEEE Transactions on}
  \bibinfo{volume}{6}~(\bibinfo{number}{1}) (\bibinfo{year}{1998})
  \bibinfo{pages}{2--14}.

\bibitem[{Bayo(1987)}]{fem}
\bibinfo{author}{E.~Bayo}, \bibinfo{title}{A finite-element approach to control
  the end-point motion of a single-link flexible robot},
  \bibinfo{journal}{Journal of Robotic systems}
  \bibinfo{volume}{4}~(\bibinfo{number}{1}) (\bibinfo{year}{1987})
  \bibinfo{pages}{63--75}.

\bibitem[{Ortega et~al.(2016)Ortega, Donaire, and Romero}]{ORTDONROM}
\bibinfo{author}{R.~Ortega}, \bibinfo{author}{A.~Donaire},
  \bibinfo{author}{J.~Romero}, \bibinfo{title}{Lecture Notes in Control and
  Information Sciences}, chap. \bibinfo{chapter}{Passivity--based control of
  mechanical systems}, \bibinfo{publisher}{Springer Berlin/Heidelberg},
  \bibinfo{year}{2016}.

\bibitem[{Donaire et~al.(pear)Donaire, Mehra, Ortega, Satpute, Romero, Kazi,
  and Singh}]{DONetal}
\bibinfo{author}{A.~Donaire}, \bibinfo{author}{R.~Mehra},
  \bibinfo{author}{R.~Ortega}, \bibinfo{author}{S.~Satpute},
  \bibinfo{author}{J.~Romero}, \bibinfo{author}{F.~Kazi},
  \bibinfo{author}{N.~Singh}, \bibinfo{title}{Shaping the energy of mechanical
  systems without solving partial differential equations},
  \bibinfo{journal}{Automatic Control, IEEE Transactions on} .

\bibitem[{Spong(1998)}]{spong}
\bibinfo{author}{M.~W. Spong}, \bibinfo{title}{Underactuated mechanical
  systems}, in: \bibinfo{booktitle}{Control problems in robotics and
  automation}, \bibinfo{publisher}{Springer}, \bibinfo{pages}{135--150},
  \bibinfo{year}{1998}.

\bibitem[{Laura et~al.(1974)Laura, Pombo, and Susemihl}]{LAUetal}
\bibinfo{author}{P.~Laura}, \bibinfo{author}{J.~Pombo},
  \bibinfo{author}{E.~Susemihl}, \bibinfo{title}{A note on the vibrations of a
  clamped-free beam with a mass at the free end}, \bibinfo{journal}{Journal of
  Sound and Vibration} \bibinfo{volume}{37}~(\bibinfo{number}{2})
  (\bibinfo{year}{1974}) \bibinfo{pages}{161--168}.

\bibitem[{Ortega and Spong(1989)}]{ORTSPO}
\bibinfo{author}{R.~Ortega}, \bibinfo{author}{M.~W. Spong},
  \bibinfo{title}{Adaptive motion control of rigid robots: A tutorial},
  \bibinfo{journal}{Automatica} \bibinfo{volume}{25}~(\bibinfo{number}{6})
  (\bibinfo{year}{1989}) \bibinfo{pages}{877--888}.

\bibitem[{Hill and Mareels(1990)}]{MARHIL}
\bibinfo{author}{D.~J. Hill}, \bibinfo{author}{I.~M. Mareels},
  \bibinfo{title}{Stability theory for differential/algebraic systems with
  application to power systems}, in: \bibinfo{booktitle}{Robust Control of
  Linear Systems and Nonlinear Control}, \bibinfo{publisher}{Springer},
  \bibinfo{pages}{437--445}, \bibinfo{year}{1990}.

\bibitem[{Khalil(2000)}]{khalil}
\bibinfo{author}{H.~Khalil}, \bibinfo{title}{Nonlinear systems},
  \bibinfo{publisher}{Prentice-Hall, New Jersey}, \bibinfo{year}{2000}.

\bibitem[{Woolsey et~al.(2004)Woolsey, Reddy, Bloch, Chang, Leonard, and
  Marsden}]{WOOLetal}
\bibinfo{author}{C.~Woolsey}, \bibinfo{author}{C.~K. Reddy},
  \bibinfo{author}{A.~M. Bloch}, \bibinfo{author}{D.~E. Chang},
  \bibinfo{author}{N.~E. Leonard}, \bibinfo{author}{J.~E. Marsden},
  \bibinfo{title}{Controlled Lagrangian systems with gyroscopic forcing and
  dissipation}, \bibinfo{journal}{European Journal of Control}
  \bibinfo{volume}{10}~(\bibinfo{number}{5}) (\bibinfo{year}{2004})
  \bibinfo{pages}{478--496}.

\bibitem[{G{\'o}mez-Estern and van~der Schaft(2004)}]{GOMAvdS}
\bibinfo{author}{F.~G{\'o}mez-Estern}, \bibinfo{author}{A.~J. van~der Schaft},
  \bibinfo{title}{Physical damping in IDA-PBC controlled underactuated
  mechanical systems}, \bibinfo{journal}{European Journal of Control}
  \bibinfo{volume}{10}~(\bibinfo{number}{5}) (\bibinfo{year}{2004})
  \bibinfo{pages}{451--468}.

\end{thebibliography}

\appendix

\section{Details of the System Dynamics}
\lab{appa}
The functions below are the elements of the matrices $A$, $B$, $C$ and $D$ in system \eqref{EL}.
\small
%\vspace{-.2cm}
\begin{align*}
A_1 (\theta,x_e) &:=   \int_{0}^{x_e}\frac{\theta[\phi'(x)]^2}{\sqrt{1+\left[\theta\phi'(x)\right]^2}}\mathrm{d} x,\\
  A_2 (\theta,x_e) &:=\sqrt{1+\left[\theta\phi'(x_e)\right]^2},\; A_3 (\theta,x_e) := \frac{2\theta[\phi'(x_e)]^2}{\sqrt{1+\left[\theta\phi'(x_e)\right]^2}},\\ A_5 (\theta,x_e)  &:=   \int_{0}^{x_e}\frac{[\phi'(x)]^2}{\left\lbrace1+\left[\theta\phi'(x)\right]^{2}\right\rbrace^{\frac{3}{2}}}\mathrm{d} x. \\
 B_1 (\theta,x_e) &:= EI\int_{0}^{x_e}\frac{\theta[\phi''(x)]^2\left\lbrace 1- 2[\theta\phi'(x)]^2\right\rbrace}{\left\lbrace1+\left[\theta\phi'(x)\right]^2\right\rbrace ^4}\mathrm{d} x,\\ B_2 (\theta,x_e) &:= \hal  \frac{EI[\theta\phi''(x_e)]^2}{\left\lbrace1+\left[\theta\phi'(x_2)\right]^2\right\rbrace ^3}+D_3g.\\
 C_1 (x_e) &:=2D_3\phi(x_e)\phi'(x_e), \; \; C_2 (x_e) :=D_3\phi'(x_e). \\
   D_1 (x_e) &:=\rho \mathcal{A}_0 \int^L_0 [\phi(x)]^2 \mathrm{d} x+D_3[\phi(x_e)]^2,\\    D_2 (x_e) &:= D_3\phi(x_e)+\rho \mathcal{A}_0 \int^L_0 \phi(x) \mathrm{d}x \\ D_4&:= D_3+M_c+\rho \mathcal{A}_0 L.
\end{align*}
\normalsize
\section{Experimental Implementation }
\lab{appb}

Figure~\ref{pendulum} shows the picture of the setup used for experimental implementation. A fatigue resistant Cu-Be alloy material is used for fabrication of the beam. Cart is guided by a rail and driven through a toothed belt driven by a motor (Maxon Motor AG: 236670). An encoder reads the position $z$ of the motor and hence the cart. An H-bridge amplifier (Nex Robotics Hercules 36V,15A) is used to drive the motor. Strain gauges (TML Tokyo Sokki Kenkyujo Co.: FLA-5-11) in full bridge configuration along with an amplifier (DataQ Instruments 5B38-02) are used for feedback $\theta$. The derivatives $\dot\theta$ and $\dot z$ are computed numerically using a digital derivative filter. Interfacing of the motor, strain amplifier, and encoder is done with data acquisition system ds 1104 from dSPACE GmbH via PWM, DAC, and encoder interfaces. Careful horizontal levelling of the cart and rail, and meticulous adjustment of the beam and the center of gravity of the tip mass is carried out to make sure that the unstable equilibrium is perfectly vertical and other equilibria are symmetric about the unstable equilibrium position. Several nonlinear terms in the control law~(\ref{pidcon}) are integral function of $\theta$ with length constraint giving $\dot x_e$ as limit of integration and thus are computationally demanding to evaluate in real time. Hence a look up table arrangement is used for evaluation of these terms in real time. Appropriate signal conditioning is used to balance detrimental effects of noise on one side and filter delay on the other.

\end{document}